\documentclass[11pt]{article}
\usepackage{amsmath,amssymb,graphics}
\usepackage{bbold}
\usepackage{amsthm}
\usepackage{graphicx}
\usepackage{wasysym}
\usepackage[]{geometry}
\usepackage[all]{xy}
\usepackage{stmaryrd}
\usepackage{tabularx}
\usepackage{longtable}
\usepackage{xcolor}
\usepackage{todonotes}
\usepackage{tikz,tikz-3dplot}
\usetikzlibrary{calc,arrows,positioning}
\usepackage[T1]{fontenc}
\usepackage{url}
\usepackage[square, numbers, sort&compress, merge]{natbib} %natbib bibliography
\usepackage{authblk}
\usepackage{lastpage}
\usepackage{fancyhdr}
\usepackage[linktocpage=false,pdfborder={0 0 0},citecolor=blue,pdfstartview={FitH}]{hyperref}

\newcommand{\ket}[1]{|#1\rangle}
\newcommand{\bra}[1]{\langle #1|}

\theoremstyle{plain}% default
\newtheorem{thm}{Theorem}[section]
\newtheorem{lem}[thm]{Lemma}
\newtheorem{prop}[thm]{Proposition}
\newtheorem{cor}{Corollary}[section]

\theoremstyle{definition}
\newtheorem{defn}{Definition}[section]

\theoremstyle{remark}

\numberwithin{equation}{section}

\title{Poincar\'e polynomials for Abelian symplectic quotients of pure $r$-qubits via wall-crossings}

\author[1,2]{Saeid Molladavoudi \thanks{E-mail: \href{mailto:smollada@uottawa.ca}{\texttt{smollada@uottawa.ca}}}}
\author[2,3]{Hishamuddin Zainuddin \thanks{E-mail: \href{mailto:hisham@upm.edu.my}{\texttt{hisham@upm.edu.my}}}}
\affil[1]{\textit{\small Department of Mathematics and Statistics, University of Ottawa, 585 King Edward, Ottawa ON, K1N 6N5, Canada}}
\affil[2]{\textit{\small Laboratory of Computational Sciences \& Mathematical Physics, Institute For Mathematical Research, Universiti Putra Malaysia, 43400 UPM Serdang, Selangor, Malaysia}}
\affil[3]{\textit{\small Department of Physics, Faculty of Science, Universiti Putra Malaysia, 43400 UPM Serdang, Selangor, Malaysia}}
\date{\vspace{-2ex}\small{Dedicated to Professor Syed Twareque Ali}}

\begin{document}
\maketitle

\begin{abstract} In this paper, we compute a recursive wall-crossing formula for the Poincar\'e polynomials and Euler characteristics of Abelian symplectic quotients of a complex projective manifold under a special effective action of a torus with non-trivial characters. An analogy can be made with the space of pure states of a composite quantum system containing $r$ quantum bits under action of the maximal torus of Local Unitary operations.
\end{abstract}

\section{Introduction} \label{introduction}

In physics, the space of pure multi-particle quantum states can be represented by a K\"ahler manifold $(M=\mathbb{P}_n,\omega_{FS})$, where $\omega_{FS}$ denotes the Fubini-Study symplectic form \cite{ashtekar1995,brody2001}. The group of Local Unitary (LU) operations is the symmetry group of a multi-partite quantum system containing $r$-isolated quantum particles. Particularly, when the system contains $r$ quantum bits (qubits) the LU group is represented by the compact, semi-simple Lie group $K = SU(2)^{\times_r}$ acting in a Hamiltonian way on the space of pure multi-partite quantum states $(M=\mathbb{P}_n,\omega_{FS})$, such that $n+1=2^r$, and is equipped with an equivariant moment map $\mu: \mathbb{P}_n \rightarrow \mathfrak{k}^*$ \cite{benvegnu2004,sawicki2011a,saeid2012b}. The Lie group $K$ itself is a subgroup of Special Unitary group $SU(n+1)$.

Now, consider the torus $T^r=(S^1)^r$ acting on the K\"ahler manifold $M=\mathbb{P}_n$, such that $2^r = n+1$. The torus $T^r$ is in fact the maximal torus of the LU group $K=SU(2)^r$. In fact, $(\mathbb{P}_n,\omega_{FS},T^r,\mu_T)$ is a Hamiltonian $T$-manifold, where $\mu_T:\mathbb{P}_n \rightarrow \mathfrak{t}^*_r$ denotes an equivariant moment map whose values $\xi=\mu_T(p)$, for $p \in \mathbb{P}_n$, are collections of diagonal elements of the \textit{reduced density matrices}, or sets of fixed phases for $r$-qubits, and by fixing them one can study geometrical and topological invariants of the associated symplectic reduced spaces $M_{\xi}= \mu_T^{-1}(\xi)/T$, as $\xi$ varies in $\mu_T(\mathbb{P}_n) \equiv \Delta$. Recall that $\mu_T(M)$ is the convex hull of the finite set of points $\mu_T(M^T)$, where $M^T$ denotes the fixed point set of the $T$-action on $M$. Moreover, for a regular value $\xi \in \Delta_{\text{reg}}$, the associated symplectic reduced space $M_{\xi}= \mu_T^{-1}(\xi)/T$ possesses at most orbifold singularities \cite{kirwan1984,kirwan1985}.

In this paper we are interested in utilizing a recursive wall-crossing formula for constructing topological invariants, specifically the Poincar\'e polynomials and the Euler characteristics of symplectic reduced spaces $M_{\xi}$ as $\xi$ crosses codimension-1 walls between two adjacent chambers as proposed in \cite{metzler1997,metzler2000}. Since for any regular value $\xi \in \Delta_{\text{reg}}$, the topological invariants of symplectic quotient $M_{\xi}$ will be a function of only chambers, the problem  becomes combinatorial in essence.

The outline of the paper is as follows. In section \ref{recursive} we review the recursive wall-crossing formula proposed in \cite{metzler2000} to find Poincar\'e polynomials for symplectic quotients of Hamiltonian torus $T$-manifolds. Then, in section \ref{recursive-poincare}, we obtain the recursive wall-crossing Poincar\'e polynomial and Euler characteristic for a regular value of the torus moment map $\mu_T$ for the Hamiltonian manifold $(\mathbb{P}_n,\omega_{FS},\mu_T,T^r)$, followed by detailed examples for cases $r=2,3$ in sub-section \ref{examples} and finally in section \ref{summary} we summarize the results and provide outlooks for future research.

\section{Recursive Wall-crossing Invariants} \label{recursive}

Consider $(M,\omega,\mu_T,T)$ as a Hamiltonian torus $T$-manifold such that $\mu_T: M \rightarrow \mathfrak{t}^*$ is the equivariant moment map. Given a point $\xi \in \mathfrak{t}^*$, the symplectic reduced space $M_{\xi}$ is defined as follows
\begin{equation} \label{regular-symplectic-quotient}
M_{\xi}=\mu_T^{-1}(\xi)/T\equiv M \sslash_{\xi}T.
\end{equation}
If $\xi$ is a regular value of $\mu_T$, then $M_{\xi}$ is called the Marsden-Weinstein symplectic reduced space \cite{marsden1974} and possesses at most finite quotient (orbifold) singularities and the induced symplectic structure $\omega_{\xi}$ on $M_{\xi}$ is defined as
\[
\pi^* \, \omega_{\xi} = i^* \omega,
\]
where $\pi: \mu_T^{-1}(\xi) \rightarrow M_{\xi}$ is the projection and $i: \mu_T^{-1}(\xi) \hookrightarrow M$ is the inclusion of the level set. By the convexity theorem  \cite{atiyah1982,guillemin1982}, the image $\Delta = \mu_T(M)$ is the convex hull of the image of the fixed point set $M^T$, so a convex polytope itself. The set of regular values of the moment map $\Delta_{\text{reg}}$ is then a finite union of open convex sub-polytopes (chambers) \cite{guillemin1982}. 
\begin{prop} \label{prop1}
 \cite{dh1982} Let $(M,\omega,\mu_T,T)$ be a compact and connected Hamiltonian $T$-manifold and consider $\xi, \eta \in \Delta_{\text{reg}}$ in the same chamber. Then the two symplectic quotients $M_{\xi}$ and $M_{\eta}$ are diffeomorphic.
\end{prop}
According to Proposition \ref{prop1}, any topological invariant of the symplectic quotient $M_{\xi}$, when $\xi \in \Delta_{\text{reg}}$, is only a function of the combinatorial properties of the chamber that $\xi$ is in.
However, an invariant of a symplectic reduced space changes when the moment value $\xi$ crosses boundaries or critical walls, which can be expressed in terms of a ``wall-crossing'' formula. There exist wall-crossing formulas for invariants of symplectic quotients, such as the change in the cohomology class $[\omega]_{\text{red}}$ \cite{gs1989} or the Duistermaat-Heckman polynomial \cite{gls1996}. 

In \cite{metzler1997,metzler2000}, Metzler proposed a recursive wall-crossing formula for topological invariants of symplectic reduced spaces of Hamiltonian $T$-manifolds by using an object called a ``weighted X-ray'' encoding all the fixed point data of the torus $T$ and sub-tori actions. To investigate the torus action, he used the \textit{(infinitesimal) orbit-type} variant of Tolman's \textit{X-ray} \cite{tolman1998}, which not only includes the fixed point data, but also encodes the fixed point sets of sub-tori $H \subset T$. 

Recall that the \textit{orbit-type stratification} of a manifold $M$, with an action of a torus $T$, is given by the decomposition of $M$ into a finite set of the connected components of strata $M_H = \left\{ p \in M \, | \, T_p = H \right\}$, where $T_p$ is the \textit{stabilizer} (or isotropy) subgroup of point $p \in M$ and $H$ is a subtorus $H \subset T$. Equivalently, an infinitesimal orbit-type stratification is determined by the infinitesimal stabilizers, i.e. the connected components of $M_{\mathfrak{h}} = \left\{ p \in M \, | \, \mathfrak{t}_p = \mathfrak{h} \right\}$, where $\mathfrak{t}_p$ and $\mathfrak{h}$ are Lie algebras of the subgroups $T_p$ and $H$, respectively \cite{ginzburg2002}. Let $\left\{X_1, \cdots, X_m \right\}$ denotes the finite set of connected components of infinitesimal orbit-type strata, with corresponding stabilizers $T_1,\cdots,T_m$, i.e. $M=\bigcup_{j=1}^m X_j$. For each $X_j$, let $F_j$ denotes its closure, i.e. $\bar{X_j}=F_j$. Then, each $F_j \in \mathcal{F} = \left\{ F_1, \cdots , F_m \right\}$ is a connected component of the fixed point set $M^{T_j}$ of the subtorus $T_j \subset T$, as the stabilizer of a generic point in $F_j$. By the equivariant Darboux theorem, each $F_j$ is a symplectic manifold which is acted upon effectively and in a Hamiltonian fashion by the quotient torus $T/T_j$ \cite{guillemin1984}. The corresponding moment map $\mu_{T/T^j}$ is obtained by the restriction of the torus--$T$ moment map $\mu_T$ to the sub-manifold $F_j$. Then, each $\mu_{T/T^j}(F_j)$ is a convex sub-polytope on its own right with $\text{dim} \, (\mu_{T/T^j}(F_j)) = d - d_j$, where $d = \text{dim} \, T$ and $d_j = \text{dim} \, T^j$. For instance, if $T=T^j$, then $\mu_{T/T^j}(F_j)$ defines a point in $\mathfrak{t}^*$. 

\begin{defn} \label{x-ray-defn}
 \cite{tolman1998,metzler2000} An \textbf{X-ray} of a compact Hamiltonian $T$-manifold $(M,\omega,\mu_T,T)$ is a family of convex sub-polytopes $\left\{ \mu_{T/T^j}(F_j) \right\} \subset \mu_T(M)$ indexed by the set $\mathcal{F}$, which is a partially ordered set under inclusion. A wall of an X-ray is denoted by the pair $(F_j,\mu_{T/T^j}(F_j))$. 
\end{defn}

The \textit{vertices} of an X-ray, which are single points in $\mathfrak{t}^*$, are images under the moment map $\mu_T$ of $T$-fixed points components, while the \textit{lines} are images of strata corresponding to circles $S^1$ in $T$, i.e. they are images under the moment map $\mu_{T^d/T^{d-1}}$ of $T^{d-1}$-fixed points components, and so on for higher dimensional walls.

Now, let $\left\{ w_{l,j} \in \mathfrak{t}^* \right\}$ denote the weights of the infinitesimal action of torus $T$ on the normal bundle $N_j$ to each $F_j$, for $l = 1, \cdots , \text{rank}(N_j)/2$. The reason is that the normal bundle $N_j$ to each $F_j$ has a natural isotopy class of complex structures induced from the natural isotopy class of compatible, invariant almost complex structures on $M$ because of the $T$-invariant symplectic form \cite{mcduff1995}. Therefore, for all $p \in M^{T^j}$, the weights of the infinitesimal action of $T^j$ on $T_pM$ are well-defined vectors in $\mathfrak{t}^*$, which allows us to define the weights of the infinitesimal torus $T$-action on the normal bundle $N_j$ to each $F_j$. Of course, for each $N_j$ there are additional zero weights corresponding to $T_pF_j$, but this will allow us to record the weight data package and also the multiplicities of weights pointing in the same direction more consistently \cite{metzler1997}. 
 
The fixed point data that Metzler's recursive wall-crossing formula depends on, comprise of the fixed point connected components $F_j$, the weights $\left\{ w_{l,j} \right\}$ and the points $\mu_T(F_j) \in \mu_T(M)$, where the data involved in dimension $d$ wall-crossing is obtained by the wall-crossing procedure in dimension $d-1$. Hence, a \textbf{weighted X-ray} of a Hamiltonian $T$-manifold $(M,\omega,\mu_T,T)$ is an X-ray, as in Definition \ref{x-ray-defn}, together with the weights $\left\{ w_{l,j} \right\}$ assigned to each stratum $F_j$. There are several properties of the weighted X-rays of Hamiltonian $T$-manifolds, but only the following ones are in particular relevant in this paper: first, the weights attached to the $T$-fixed point components, as vertices of an X-ray, lie along \textit{one}-dimensional walls of the corresponding X-ray and it is important to keep record of multiplicities of weights pointing in the same direction along the lines connecting vertices of an X-ray, and second, the weight data at these vertices often completely determine the X-ray \cite{metzler2000}. 

Recall that the set of regular values of the moment map $\mu_T$ is the complement of the union of the walls of an X-ray, i.e. 
\[ 
\Delta_{\text{reg}} = \mu_T(M) \backslash \bigcup_j \mu_{T/T^j}(F_j),
\]
which itself is an open set with a finite number of components, so called chambers. Equivalently, for each $F_j$ as an effective Hamiltonian $T/T^j$ space, we can define the set of regular values for the associated moment map $\mu_{T/T^j}$ as follows:
\begin{equation} \label{sub-chambers-eq}
\Delta^{(j)}_{\text{reg}} = \mu_{T/T^j}(F_j) \backslash \bigcup_{F_k < F_j} \mu_{T/T^k}(F_k),
\end{equation}
which is a relatively open set with finite number of components, as convex sub-polytopes or \textit{sub-chambers} of the X-ray. In other words, every point $\xi \in \mathfrak{t}^*$ in the moment polytope $\mu_T(M) \subset \mathfrak{t}^*$ is a regular value for some Hamiltonian action, for instance the vertices are regular values, since they correspond to the $T/T$ torus, which is trivial with a zero-dimensional Lie algebra, and the points in the one-dimensional walls (or the lines) that are connecting the vertices are also regular values of the moment map $\mu_{T^d/T^{d-1}}  = \mu_{S^1}$, and so on for higher dimensional walls in the moment polytope. More precisely, consider a stratum $F_j \in \mathcal{F}$ in such a way that the torus $T^j$ is the stabilizer subgroup of a generic point in $F_j$ and let $\xi \in \Delta^{(j)}_{\text{reg}} \subset \mu_{T/T^j}(F_j)$ as a regular value of the moment map $\mu_{T/T^j}$. Then, the corresponding regular symplectic reduction can be obtained as
\begin{equation} \label{regular-sub-reduction}
M_{\xi}^{(j)} = \left( \mu^{-1}_{T/T^j}(\xi) \cap F_j \right)/ \left( T/T^j \right).
\end{equation}
However, the symplectic reduction \eqref{regular-sub-reduction} has to be distinguished from the singular symplectic reduction $M_{\xi} = \mu_T^{-1}(\xi)/T$, when $\xi \not\in \Delta_{\text{reg}}$, that is in general a symplectic stratified space in which $M_{\xi}^{(j)}$ is a stratum. In addition, we have to specify which strata $F_j$ we are restricting for the symplectic reduction purpose, since there may be more than one strata whose moment images contain $\xi$, for instance if $F_k \subset F_j$, then $\mu_{T/T^k}(F_k) \subset \mu_{T/T^j}(F_j)$ but not vice versa. In fact, the recursive wall-crossing formula \eqref{recursive-wall-crossing-eq} below, takes into account all of these lower-dimensional ``sub-reductions'', but before that we need the following lemma that describes how two sub-chambers of an X-ray can meet.
\begin{lem} \label{sub-chambers-meet-lem}
\cite{metzler1997,metzler2000} Consider the weighted X-ray of a Hamiltonian torus-$T$ manifold $(M,\omega,\mu_T,T)$ and let $(F_j,\mu_{T/T^j}(F_j))$ be a wall in such a way that $(F_j,\Delta^{(j,1)}_{\text{reg}})$ and $(F_j,\Delta^{(j,2)}_{\text{reg}})$ are two adjacent sub-chambers from Eq. \eqref{sub-chambers-eq}, whose closures intersect in a codimension-1 subwall (or principal subwall) in $\mu_{T/T^j}(F_j)$. Then,
\begin{enumerate}
\item $\delta:=\bar{\Delta}^{(j,1)}_{\text{reg}} \cap \bar{\Delta}^{(j,2)}_{\text{reg}}$ is a convex polytope and let $\text{Aff}(\delta)$ denotes its affine span.
\item For each stratum $F_k \subset F_j$ such that $\mu_{T/T^k}(F_k)$ is a principal subwall in $\text{Aff}(\delta)$ and $\mu_{T/T^k}(F_k) \cap \delta \neq \emptyset$, then $\delta$ lies entirely in one sub-chamber $(F_k,\Delta^{(k)}_{\text{reg}})$ of $F_k$.
\item For a non-overlapping weighted X-ray (i.e. a weighted X-ray in which for two different strata their moment images do not overlap), $\delta$ lies in a unique sub-chamber $\Delta^{(k)}_{\text{reg}}$ of a unique principal subwall $\mu_{T/T^k}(F_k)$ of $\mu_{T/T^j}(F_j)$.
\end{enumerate}
\end{lem}

Now, consider the following scenario in which $F_j \in \mathcal{F}$ contains two sub-chambers $(F_j,\Delta^{(j,1)}_{\text{reg}})$ and $(F_j,\Delta^{(j,2)}_{\text{reg}})$, whose $\delta:=\bar{\Delta}^{(j,1)}_{\text{reg}} \cap \bar{\Delta}^{(j,2)}_{\text{reg}}$ contains two or more principal subwalls $\left\{ (F_k,\mu_{T/T^k}(F_k)) \right\}_k$ with a common $\text{Aff}(\delta)$. According to the Lemma \ref{sub-chambers-meet-lem}, in every principal subwall $(F_k,\mu_{T/T^k}(F_k))$ a \textit{unique} sub-chamber $(F_k,\Delta^{(k)}_{\text{reg}})$ separates $\Delta^{(j,1)}_{\text{reg}}$ and $\Delta^{(j,2)}_{\text{reg}}$ into two components, denoted by $\Lambda_1$ and $\Lambda_2$, since $(F_k,\mu_{T/T^k}(F_k))$ is codimension-1 subwall in $(F_j,\mu_{T/T^j}(F_j))$, for every $k$. Let $\Pi:\text{Lin}\left( \mu_{T/T^j}(F_j) \right) \rightarrow \mathbb{R}$ denotes a map from the unique linear subspace $\text{Lin}\left( \mu_{T/T^j}(F_j) \right)$ parallel to the affine span $\text{Aff} \left( \mu_{T/T^j}(F_j) \right)$ onto the $\text{Lin} \left( \mu_{T/T^j}(F_j) \right) / \text{Lin}( \delta ) \cong \mathbb{R}$, and
\begin{equation}
w^{(j)}_{l,k} = w_{l,k} \cap \text{Lin}\left( \mu_{T/T^j}(F_j) \right) / \text{Lin}(\delta),
\end{equation}
where $w^{(j)}_{l,k}$ are the weights of infinitesimal torus $T$-action on the normal bundle $N_k$ to the stratum $F_k$, which lie along the wall $(F_j,\mu_{T/T^j}(F_j))$. Then, we can define
\begin{eqnarray} \label{positive-negative-weights}
b_k & = & \text{number of weights} \, \left\{ w^{(j)}_{l,k} \right\}_l \text{lying in} \, \, \Pi(\Lambda_1), \nonumber \\
f_k & = & \text{number of weights} \, \left\{ w^{(j)}_{l,k} \right\}_l \text{lying in} \, \,  \Pi(\Lambda_2),
\end{eqnarray}
where $b_k$ ($f_k$) denotes the number of weights pointing backward (forward) during the wall-crossing procedure from $\Delta^{(j,1)}_{\text{reg}}$ to $\Delta^{(j,2)}_{\text{reg}}$. In section \ref{examples} we will show explicitly through examples how to count these weights as we cross two adjacent walls in the moment polytope.
\begin{defn} \label{recursive-invariant}
A \textbf{recursive invariant} $I$ of a weighted X-ray with values in a ring $R$ consists of two following pieces of data:
\begin{enumerate}
\item a map $I: \mathcal{Q} \rightarrow R$, where $\mathcal{Q} = \left\{ \Delta^{(j)}_{\text{reg}} \right\}_j$ is the set of all sub-chambers of the weighted X-ray, and
\item a function $C_I: \mathbb{Z} \times \mathbb{Z} \rightarrow R$ called a \textit{wall-crossing} function of $I$,
\end{enumerate}
such that they satisfy the following conditions for each wall $(F_j,\mu_{T/T^j}(F_j))$:
\begin{enumerate}
\item Given two adjacent sub-chambers denoted by $(F_j,\Delta^{(j,1)}_{\text{reg}})$ and $(F_j,\Delta^{(j,2)}_{\text{reg}})$, whose $\delta:=\bar{\Delta}^{(j,1)}_{\text{reg}} \cap \bar{\Delta}^{(j,2)}_{\text{reg}}$ contains two or more principal subwalls $\left\{ (F_k,\mu_{T/T^k}(F_k)) \right\}_k$, as described above, and let $b_k$ and $f_k$ be the number of infinitesimal weights of torus-$T$ action on the normal bundle $N_k$ in $N_j$ pointing toward $\Delta^{(j,1)}_{\text{reg}}$ and $\Delta^{(j,2)}_{\text{reg}}$ respectively. Then,
\begin{equation} \label{recursive-wall-crossing-eq}
I(\Delta^{(j,2)}_{\text{reg}}) - I(\Delta^{(j,1)}_{\text{reg}}) = \sum_{k}{C_I(b_k,f_k)I(\Delta^{(k)}_{\text{reg}})}.
\end{equation}
\item Considering a sub-chamber $(F_j,\Delta_{\text{reg}}^{(j)})$ adjacent to the boundary of $\mu_{T/T^j}(F_j)$, let $(F_k,\mu_{T/T^k}(F_k))$ be the unique sub-chamber separating $\Delta_{\text{reg}}^{(j)}$ from the exterior of $\mu_{T/T^j}(F_j)$ and let $f$ be the number of weights of $N_k$ in $N_j$ pointing into $N_j$ and $b$ be the number pointing out. Then,
\begin{equation}
I(\Delta_{\text{reg}}^{(j)}) = C_I(b,f) \, I(\Delta^{(k)}_{\text{reg}}).
\end{equation}
\end{enumerate}
\end{defn}

It is implicit that $I=0$, for points outside of the moment polytope $\Delta$. According to \cite{metzler2000}, the values of a recursive invariant $I$ are completely determined by the function $C_I$ and its values on the vertices. In fact, the values of $I$ in a $d$-dimensional sub-chamber of a moment polytope can be obtained by starting outside the polytope and crossing a finite number of walls with lower dimensions recursively.
The following theorem asserts that an invariant is recursive if it is for circle $S^1$-action:
\begin{thm} \label{circle-recursive}
\cite{metzler1997,metzler2000} Let $I$ be a topological invariant of symplectic manifolds. Assume that the X-ray invariant defined by $I$ as in definition \ref{recursive-invariant} is recursive on the class of X-rays coming from Hamiltonian circle actions. Then this invariant is recursive on all Hamiltonian X-rays, with the same wall-crossing function.
\end{thm}

In other words, Theorem \ref{circle-recursive} asserts that any wall-crossing invariant can be reduced to circle actions. By using the Theorem \ref{circle-recursive}, Metzler found in \cite{metzler1997,metzler2000} a recursive wall-crossing formula for the Poincar\'e polynomials of symplectic reduced spaces of a general Hamiltonian torus-$T$ manifold $(M,\omega)$ as follows
\begin{equation} \label{recursive-wallcrossing-formula}
P_t(M_{\xi_2}) - P_t(M_{\xi_1}) = P_t(F_j) \, C(b,f),
\end{equation}
where 
\begin{equation} \label{wall-crossing-function}
C(b,f) = \frac{t^{2b}-t^{2f}}{1-t^2} = \begin{cases}
t^{2f-2}+ t^{2f-4} + \cdots + t^{2b}, \quad f>b \\
-t^{2b-2}- t^{2b-4} - \cdots - t^{2f}, \quad b>f
\end{cases},
\end{equation}
and $P_t(F_j) \equiv P_t(M^{(j)}_{\zeta}) \in \mathbb{Z}[t]$ is the Poincar\'e polynomial for symplectic quotient $M^{(j)}_{\zeta} = (\mu_{T/T^j}^{-1}(\zeta) \cap F_j)/ (T/T^j)$, as defined in Eq. \eqref{regular-sub-reduction} for $\zeta \in \Delta_{\text{reg}}^{(j)}$, which is a regular value of the moment map $\mu_{T/T^j}$ associated with the Hamiltonian torus space $(F_j,T/T^j,\mu_{T/T^j})$. In Eq. \eqref{wall-crossing-function}, $f \, (b)$ is the number of positive (negative) weights in the normal bundle to $F_j$. Later, in section \ref{examples}, we see how the Theorem \ref{circle-recursive} will facilitate counting the number of positive and negative weights along the moment image of the normal bundle to the fixed points of the torus action. For more details on weighted X-rays and their recursive invariants the readers can refer to \cite{metzler1997,metzler2000}.

\section{Poincar\'e Polynomials for Abelian Quotients} \label{recursive-poincare}
Recall that the Local Unitary group $K=SU(2)^r$ acts on K\"ahler manifold $M=\mathbb{P}_n$, as the space of pure states of multi-particle quantum system consisting of $r$ quantum bits, each with a Hilbert space $\mathcal{H}_i = \mathbb{C}^2$, for $i=1, \cdots, r$, and $\mathcal{H} = \bigotimes_r \mathcal{H}_i$. In physics, points on the manifold $M$ are in fact positive semi-definite, trace-class and Hermitian matrices $\rho$, which can be written as tensor product of the density matrices of the sub-systems as $\rho = \ket{\Psi}\bra{\Psi} = \ket{\psi_1}\bra{\psi_1} \otimes \ket{\psi_2}\bra{\psi_2} \otimes \cdots \otimes \ket{\psi_r}\bra{\psi_r}$, where $\ket{\psi_i} \in \mathcal{H}_i$. Let $g=(g_1,g_2, \cdots,g_r) \in K$, where $g_i \in SU(2)$, then the $K$-action on $M$ can be determined as $g \, . \, \rho = g_1 \, \ket{\psi_1}\bra{\psi_1} \, g_1^{-1} \otimes \cdots \otimes g_r \, \ket{\psi_r}\bra{\psi_r} \, g_r^{-1}$. The torus $T^r=(S^1)^r$ is the maximal torus of compact Lie group $K$, whose elements can also be written in the form 
\[
T^r \ni (t_1,t_2, \cdots , t_r) \equiv \left( \text{diag}\,(t_1,t_1^{-1}), \cdots , \text{diag}\,(t_r,t_r^{-1})\right), \quad t_i = e^{i \, \gamma_i}, \, \gamma_i \in \mathbb{R}.
\]
More precisely, the $T^r$ action on $\mathbb{P}_n$ is diagonalizable through the homomorphism $\varphi: T^r \rightarrow U(n+1)$ as follows
\begin{equation} \label{action}
(t_1,t_2, \cdots , t_r) \mapsto \left( \prod_{i=1}^{r}t_i^{a_{i,1}}, \prod_{i=1}^{r}t_i^{a_{i,2}}, \cdots, \prod_{i=1}^{r}t_i^{a_{i,n+1}}\right),
\end{equation}
for some $a_{i,j} \in \left\{ \pm 1 \right\}$, where $\sum_{k=0}^{r}{\binom{r}{k}}=2^r= n+1$. The $r\times(n+1)$ matrix $A$ contains all $a_{i,j}$s as columns as follows
\[
A= \left( \begin{matrix}
a_{1,1} & a_{1,2} & \cdots & a_{1,n+1} \\
a_{2,1} & a_{2,2} & \cdots & a_{2,n+1} \\
\vdots & \vdots & \ddots  & \vdots \\
a_{r,1} & a_{r,2} & \cdots & a_{r,n+1} \\
\end{matrix} \right)_{r\times (n+1)}.
\]

The torus $T^r$ is a subgroup of the torus $T^{n+1}$, for which the components of the moment map $\Phi: M \rightarrow \mathfrak{t}^*_{n+1}$ are \textit{perfect Morse functions} with critical set $M^{T^{n+1}}$. 
\begin{lem} \label{morse-lemma}
Let $T^r$ be a subtorus of the torus $T^{n+1}$ as described above and assume that the fixed point set $M^{T^r}$ is finite. Then $M^{T^r} = M^{T^{n+1}}$.
\end{lem}
\begin{proof}
Let $\Phi:M \rightarrow \mathfrak{t}^*_{n+1}$ be the moment map for $T^{n+1}$-action and $X \in \mathfrak{t}_{n+1}$ be a \textit{generic} element, such that $\langle \beta , X \rangle \neq 0$ for each weight $\beta \in \mathfrak{t}^*_{n+1}$ of $T^{n+1}$ action on $T_pM$, for every $p$ in the fixed point set $M^{T^{n+1}}$. Then, the component of the moment map $\phi_X$ along $X$, given by $\phi_X(p) = \langle \phi(p), X \rangle$, is a perfect Morse function with critical set $M^{T^{n+1}}$, i.e. $\sum{\text{dim}\, H^i(M)} = |M^{T^{n+1}}|$. Similarly, let $\mu_T= pr_r \circ \Phi$ be the moment map for the $T^r$-action, which can be obtained by the projection to $\mathfrak{t}^*_r$ of the moment map $\Phi$. Then for any generic value $Y \in \mathfrak{t}_r$ we obtain the component $\mu_T^Y$, which is also a perfect Morse function on $M$. Thus, $\sum{\text{dim}\, H^i(M)} = |M^{T^{n+1}}|=|M^{T^r}|$. Since evidently $M^{T^{n+1}} \subset M^{T^r}$, the sets are equal.
\end{proof}
Hence, $M^{T^r} = \left\{ [0: \cdots:z_j: \cdots:0] \in \mathbb{P}_n | \, z_j=1, \, j=1, \cdots, n+1 \right\}$. The Abelian  moment map $\mu_T: \mathbb{P}_n \rightarrow \mathfrak{t}^*_r$ for the $T^r$-action is given by
\begin{equation} \label{moment-map}
\mu_T(p) = \frac{1}{2} \sum_{j=1}^{n+1}{|z_j|^2 \alpha_j},
\end{equation}
where $\alpha_j=(a_{1,j},a_{2,j}, \cdots,a_{r,j})^T$, for $j=1, \cdots,n+1$ are weights of the representation of $T^r$ on $\mathbb{C}^{n+1}$ \cite{kirwan1984}. In particular, for the fixed points $p_j \in M^{T^r}$, the \textit{isotropy weights} $w_{p_j} = \left\{ (w_{k,p_j}) \right\}_{k=1}^{n}$ of the infinitesimal $T^r$ action on $T_pM$ are obtained by 
\begin{equation} \label{isotropy-weights}
w_{k,p_j} = \left( \begin{matrix}
a_{1,k} - a_{1,j} \\
a_{2,k} - a_{2,j} \\
\vdots \\
a_{r,k} - a_{r,j} 
\end{matrix}
\right)_{r \times 1}, \quad k \neq j.
\end{equation}
The torus $T^r$-action on $\mathbb{P}_n$ is not a Goresky-Kottwitz-MacPherson (GKM)-type action, since the isotropy weights $w_{p_j} = \left\{ (w_{k,p_j}) \right\}_{k=1}^{n}$ in Eq. \eqref{isotropy-weights}, for every $p_j \in M^T$, are not pairwise linearly independent in $\mathfrak{t}^*$, on the contrary to the $T^{n+1}$-action as a toric variety \cite{gkm1998}. The moment polytope $\mu_T(\mathbb{P}_n)$ is the hypercube 
\[
\Delta \equiv \text{conv}.\left\{(x_1,\cdots,x_r) \in \mathbb{R}^r : \, x_i=\pm 1/2 \right\},
\]
spanned by $2^r$ vertices as the image of the fixed points under the moment map \cite{sawicki2015}. The codimension-one walls divide the moment polytope $\Delta$ into subpolytopes $\Delta_{\text{reg}}$, whose interior consists of entirely regular values of the Abelian moment map $\mu_T$.

Before describing the wall-crossing procedure, let us mention that the following lemma provides us with a dimensional criterion for the symplectic reduced space associated to a regular value of the moment map:
\begin{lem} \label{dimensional-condition}
Let $(\mathbb{P}_n,\omega,\mu_T,T^r)$ be a Hamiltonian torus $T$-manifold, such that $n=2^r-1$. Then, for $\xi$ as a regular value of $\mu_T: \mathbb{P}_n \rightarrow \mathfrak{t}^*$ in the hypercube (or $r$-cube) $\mu_T(M) \subset \mathbb{R}^r$, we have $\text{dim}_{\mathbb{R}}\, (M_{\xi}) =2^{r+1}-2r-2$, where $M_{\xi}=\mu_T^{-1}(\xi)/T$.
\end{lem}
\begin{proof}
\begin{equation}
\text{dim}_{\mathbb{R}}\, (M_{\xi}) = \text{dim} \, (\mu_T^{-1}(\xi)/T) = \text{dim}_{\mathbb{R}} \, \mathbb{P}_n - \text{dim} \, \mu_T(\mathbb{P}_n) - \text{dim} \, (T^r) =2^{r+1}-2r-2.
\end{equation}
\end{proof}

Regarding definition \ref{recursive-invariant}, theorem \ref{circle-recursive} and the recursive wall-crossing formula for change in the Poincar\'e polynomials of symplectic reduced spaces $M_{\xi}$ in Eqs. \eqref{recursive-wallcrossing-formula} and \eqref{wall-crossing-function}, we can start from outside the hypercube (or $r$-cube) $\Delta$ and cross a finite number of lower dimensional walls. In fact, this gives an algorithm for calculating a recursive invariant, i.e. the Poincar\'e polynomials, on all sub-chambers of a weighted X-ray by starting from the vertices and crossing higher dimensional walls recursively \cite{metzler2000}. The key is then to find the wall-crossing function $C_I(b,f)$ each time we cross a wall. A few examples will be discussed in details in the next section \ref{examples}.

Now, consider an $i$-dimensional wall in the $r$-cube, such that $i=0, \cdots,r-1$. Then, for each $\xi_{i+1}$ in a $(i+1)$-dimensional sub-chamber, whose closure contains the mentioned $i$-dimensional wall, we can consider a corresponding exterior point $\eta_{i+1} \notin \Delta$ connected with $\xi_{i+1}$ through a line crossing the $i$-dimensional wall. This follows from the fact that any sub-chamber of dimension $i+1$ is accessible from outside the polytope by crossing a finite number of dimension $i$ walls \cite{metzler2000}. By definition, the recursive invariants for those external points are zero. Next, is to find values of the wall-crossing function $C_i(b_i,f_i)$ each time we cross an $i$-dimensional wall. The fact is that the number of positive weights along the circle action connecting $\eta_{i+1}$ and $\xi_{i+1}$ is $2^{i}$ and the number of negative weights is always zero. Namely,
\begin{equation}
C_i(b_i,f_i) = C_{i}(0,2^{i}) = 1+t^2+ \cdots + t^{2^{i+1}-2}.
\end{equation}

In other words, the recursive procedure can be described as follows:
\begin{itemize}
\item $i=0$: then $\xi_1 \in \text{$1$-dimensional boundary of } \Delta$ and $f_0=1,b_0=0$. Therefore,
\[
P_t(M_{\xi_1}) - \underbrace{P_t(M_{\eta_1})}_\textrm{$0$} = P_t(F_j) \, \underbrace{C_0(0,1)}_\textrm{$1$} = 1.
\]
Evidently, the $0$-dimensional walls or the vertices of the polytope $\Delta = \mu_T(\mathbb{P}_n)$ are the image under $\mu_T$ of the fixed points under torus $T^r$-action. Hence, we have $P_t(F_l)=1$, for $l=0, \cdots,n+1$.

\item $i=1$: then $\xi_2 \in \text{$2$-dimensional boundary of } \Delta$ and $f_1=2,b_1=0$. Therefore,
\[
P_t(M_{\xi_2}) - \underbrace{P_t(M_{\eta_2})}_\textrm{$0$} = \underbrace{P_t(M_{\xi_1})}_\textrm{$1$} \, \underbrace{C_1(0,2)}_\textrm{$1+t^2$} = 1+t^2.
\]

\item $i=2$: then $\xi_3 \in \text{$3$-dimensional boundary of } \Delta$ and $f_2=4,b_2=0$. Therefore,
\[
P_t(M_{\xi_3}) - \underbrace{P_t(M_{\eta_3})}_\textrm{$0$} = \underbrace{P_t(M_{\xi_2})}_\textrm{$1+t^2$} \, \underbrace{C_2(0,4)}_\textrm{$1+t^2+t^4+t^6$} = 1+2t^2+2t^4+2t^6+t^8.
\]
\item Repeating the same procedure recursively, for a general $i$ we will have $\xi_{i+1} \in (i+1)$-dimensional boundary of $\Delta$ and $f_i=2^{i},b_i=0$. Therefore,
\[
P_t(M_{\xi_{i+1}}) - \underbrace{P_t(M_{\eta_{i+1}})}_\textrm{$0$} = \prod_{m=0}^{i}{1+t^2+t^4+ \cdots + t^{2^{m+1}-2}}.
\]
\end{itemize}
Hence, for $\xi$ as a regular value of the moment map $\mu_T$, the Poincar\'e polynomial for the associated symplectic reduced space $M_{\xi}$ can be obtained recursively from the following wall-crossing formula
\begin{equation} \label{recursive-wall-crossing-formula-abelian}
P_t(M_{\xi}) = \prod_{i=0}^{r-1}\left( \sum_{j=0}^{2^i-1}{t^{2j}} \right), 
\end{equation}  
satisfying both the Poincar\'e duality and the dimensional constraint in lemma \ref{dimensional-condition}. 

By replacing $t=-1$ we will obtain a wall-crossing function for the Euler characteristic as follows:
\begin{equation} \label{euler-characteristic-function}
\omega_{\chi}(b,f) = b-f.
\end{equation}
\begin{cor} \label{euler-characteristic}
The Euler characteristic of the symplectic quotients $M_{\xi}$, for a regular $\xi \in \Delta_{\text{reg}}$, is given by
\begin{equation} \label{euler-characteristic-wall-crossing}
\chi(M_{\xi}) = \prod_{i=0}^{r-1}{2^i}.
\end{equation}
\end{cor}

\subsection{Examples} \label{examples}

In this sub-section, we describe the recursive wall-crossing procedure for the following examples: 

\begin{enumerate}
\item [(a)] Let's consider the torus $T^2=S^1 \times S^1$ acting on the complex projective manifold $\mathbb{P}_3$ via the homomorphism $\varphi: T^2 \rightarrow U(4)$ as follows
\begin{equation} \label{action-example1}
(t_1,t_2) \mapsto \left( t_1t_2, t_1t_2^{-1}, t_1^{-1}t_2,t_1^{-1}t_2^{-1}\right).
\end{equation}
The moment polytope $\mu_T(\mathbb{P}_3)$, which is shown in Fig. \ref{2-cube}, is a square as a convex polytope $\Delta_2 \equiv \text{conv}.\left\{(x_1,x_2) \in \mathbb{R}^2 : \, x_i=\pm 1/2 \right\}$ spanned by $4$ vertices as the image of fixed points under the moment map. Since the complex dimension of the manifold $M=\mathbb{P}_3$ is three, at each vertex there exist three weights emanating from them towards other fixed points.
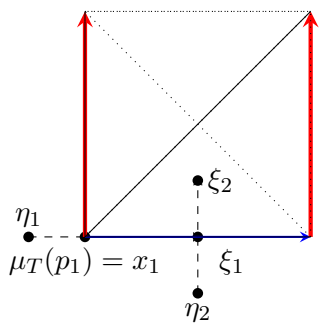
\begin{figure}[h]
\centering
\begin{tikzpicture}[
  vector/.style={ultra thick,black,>=stealth,->},
  atom/.style={blue}, x=3cm,y=3cm
  ]
    \coordinate (p0) at (0,0);
    \node[below] at (p0) {$\mu_T(p_1) =x_1$};
    \fill[black] (p0) circle (2pt);
    \coordinate (p1) at (1,0);
    \coordinate (p2) at (0,1);
    \coordinate (p3) at (1,1);
    
    \coordinate (a1) at (-0.25,0);
    \node[above] at (a1) {$\eta_1$};
    \fill[black] (a1) circle (2pt);
    
    \coordinate (a2) at (0.5,0);
    \node[below=.5em of a2] (dummy) {};
    \node[right=.005 of dummy] {$\xi_1$};
    \fill[black] (a2) circle (2pt);    
    \draw[dashed,->] (a1) -- (a2);
    
        \coordinate (a3) at (0.5,-0.25);
        \node[below] at (a3) {$\eta_2$};
        \fill[black] (a3) circle (2pt);
    
        \coordinate (a4) at (0.5,0.25);
        \node[right] at (a4) {$\xi_2$};
        \fill[black] (a4) circle (2pt);
        \draw[dashed,->] (a3) -- (a4);
        
    \draw[thick,blue,>=stealth,->] (p0) -- (p1);
    \draw[ultra thick,red,>=stealth,->] (p0) -- (p2);
    \draw[ultra thick,red,>=stealth,->] (p1)--(p3);        
    
    \draw[black] (p0) -- (p1);
    \draw[black] (p0) -- (p2);
    \draw[black] (p0) -- (p3);
    \draw[dotted] (p0) -- (p3);
    \draw[dotted] (p1) -- (p2);
    \draw[dotted] (p3) -- (p2);
    \draw[dotted] (p1) -- (p3);
    \draw[dotted] (p2) -- (p3);
\end{tikzpicture}
\caption{$\mu_T(\mathbb{P}_3)\equiv \Delta_2$, for $M = \mathbb{P}_3$ with $T^2$ Hamiltonian action}
\label{2-cube} 
\end{figure}

Regarding the Eq. \eqref{recursive-wall-crossing-formula-abelian}, or following the recursive procedure illustrated in Fig. \ref{2-cube} and described in the previous section \ref{recursive-poincare}, one can find
\begin{equation} \label{poincare-example1}
P_t(M_{\xi_2}) = \prod_{i=0}^{1}\left( \sum_{j=0}^{2^i-1}{t^{2j}} \right) = 1+t^2.
\end{equation}
Considering corollary \ref{euler-characteristic} and Eq. \eqref{euler-characteristic-wall-crossing}, the Euler characteristic of $M_{\xi_2}$ is equal to $2$. 

\item[(b)] Now, consider the torus $T^3$ acting on the K\"ahler manifold $\mathbb{P}_7$ via the homomorphism $\varphi: T^3 \rightarrow U(8)$ as follows
\begin{equation} \label{action-example}
(t_1,t_2,t_3) \mapsto \left( \prod_{i=1}^{3}t_i^{a_{i,1}}, \prod_{i=1}^{3}t_i^{a_{i,2}}, \cdots, \prod_{i=1}^{3}t_i^{a_{i,8}}\right),
\end{equation}
for some $a_{i,j} \in \left\{ \pm 1 \right\}$, where $8 = \sum_{k=0}^{3}{\binom{3}{k}}=2^3$. The moment polytope $\mu_T(\mathbb{P}_7)$ (depicted in the following Fig. \ref{3-cube-example}) is the $3$-cube 
\[
\Delta_3 \equiv \text{conv}.\left\{(x_1,x_2,x_3) \in \mathbb{R}^3 : \, x_i=\pm 1/2 \right\},
\]
spanned by $2^3$ vertices as the images of fixed points under the moment map. The codimension-one walls divide the moment polytope $\Delta$ into subpolytopes $\Delta_{\text{reg}}$, whose interior consists of entirely regular values of the Abelian moment map $\mu_T$.
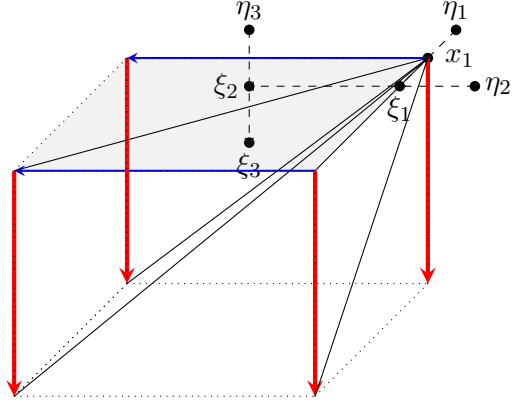
\begin{figure}[h]
\centering
\begin{tikzpicture}[
  vector/.style={ultra thick,black,>=stealth,->},
  atom/.style={blue}, x=4cm,y=3cm,z=1.5cm
  ]
    \coordinate (p0) at (0,0,0);
    \coordinate (p1) at (1,0,0);
    \coordinate (p2) at (0,1,0);
    \coordinate (p3) at (0,0,1);

    \coordinate (p4) at (1,1,0);
    \coordinate (p5) at (1,0,1);
    \coordinate (p6) at (0,1,1);
    \coordinate (p7) at (1,1,1);
    \node[right=.25em] at (p7) {$x_1$};
    \fill[black] (p7) circle (2pt);
    
    \coordinate (a1) at (1,1,1.25);
    \node[above] at (a1) {$\eta_1$};
    \fill[black] (a1) circle (2pt);

    \coordinate (a2) at (1,1,0.75);
    \node[below] at (a2) {$\xi_1$};
    \fill[black] (a2) circle (2pt);    
    \draw[dashed,->] (a1) -- (a2);
    
    \coordinate (a3) at (1.25,1,0.75);
    \node[right] at (a3) {$\eta_2$};
    \fill[black] (a3) circle (2pt);

    \coordinate (a4) at (0.5,1,0.75);
    \node[left] at (a4) {$\xi_2$};
    \fill[black] (a4) circle (2pt);
    \draw[dashed,->] (a3) -- (a4);

    \coordinate (a5) at (0.5,1.25,0.75);
    \node[above] at (a5) {$\eta_3$};
    \fill[black] (a5) circle (2pt);

    \coordinate (a6) at (0.5,0.75,0.75);
    \node[below] at (a6) {$\xi_3$};
    \fill[black] (a6) circle (2pt);
    \draw[dashed,->] (a5) -- (a6);
    
\draw[thick,blue,>=stealth,->] (p7) -- (p6);
\draw[thick,blue,>=stealth,->] (p4) -- (p2);
\draw[ultra thick,red,>=stealth,->] (p7) -- (p5);
\draw[ultra thick,red,>=stealth,->] (p2) -- (p0);
\draw[ultra thick,red,>=stealth,->] (p4) -- (p1);
\draw[ultra thick,red,>=stealth,->] (p6) -- (p3);

\fill[gray,opacity=0.1] (p7) -- (p6) -- (p2) -- (p4);

    \draw[dotted] (p0) -- (p1);
    \draw[dotted] (p0) -- (p2);
    \draw[dotted] (p0) -- (p3);
    \draw[dotted] (p0) -- (p3);
    \draw[dotted] (p3) -- (p5);
    \draw[dotted] (p3) -- (p6);
    \draw[dotted] (p6) -- (p2);
    \draw[dotted] (p5) -- (p1);
    \draw[black] (p7) -- (p4);
    \draw[dotted] (p5) -- (p7);
    \draw[dotted] (p1) -- (p4);
    \draw[dotted] (p6) -- (p7);
    \draw[dotted] (p2) -- (p4);
    \draw[black] (p7) -- (p1);
    \draw[black] (p7) -- (p2);
    \draw[black] (p7) -- (p3);
    \draw[black] (p7) -- (p0);

\end{tikzpicture}
\caption{$\mu_T(\mathbb{P}_7)\equiv \Delta_3$, for $M = \mathbb{P}_7$ with $T^3$ Hamiltonian action}
\label{3-cube-example} 
\end{figure}

As shown in Fig. \ref{3-cube-example}, by following the recursive procedure discussed in the previous section \ref{recursive-poincare}, or simply by using the wall-crossing formula in Eq. \eqref{recursive-wall-crossing-formula-abelian}, one can find
\begin{equation} \label{poincare-example2}
P_t(M_{\xi_3}) = \prod_{i=0}^{2}\left( \sum_{j=0}^{2^i-1}{t^{2j}} \right)= 1+2t^2+2t^4+2t^6+t^8.
\end{equation}
Considering corollary \ref{euler-characteristic} and the Eq. \eqref{euler-characteristic-wall-crossing}, the Euler characteristic of $M_{\xi_3}$ is equal to $8$.

\item[(c)] Repeating the same procedure computationally, we can find the associated Betti numbers and the Euler characteristic $\chi$ of the Abelian symplectic quotients $M_{\xi}=\mu_T^{-1}(\xi)/T$, for $r=4,\cdots,8$, which are summarized in the following table:
\begin{center}
\begin{longtable}{|c|p{3.65in}|c|c|}

\hline

\textbf{$r$} & \textbf{List of even-degree Betti numbers of $\mu_T^{-1}(\xi)/T$} & \textbf{$\chi(M_{\xi})$} & dim$_{\mathbb{R}}(M_{\xi})$\\

\hline \hline

4 & \{1,3,5,7,8,8,8,8,7,5,3,1\} & 64 & 22\\

\hline

5 & \{1,4,9,16,24,32,40,48,55,60,63,64,64,64,64,64, \newline 63,60,55,48,40,32,24,16,9,4,1\} & 1024 & 52\\

\hline

6 & \{1,5,14,30,54,86,126,174,229,289,352,416,480, \newline 544,608,672,735,795,850,898,938,970,994,1010,1019, \newline 1023,1024,1024,1024,1024,1024,1024,1023,1019,1010,\newline 994,970,938,898,850,795,735,672,608,544,480,416,  \newline 352,289,229,174,126,86,54,30,14,5,1\} & 32768 & 114 \\

\hline

7 & \{1,6,20,50,104,190,316,490,719,1008,1360,1776,  \newline 2256,2800,3408,4080,4815,5610,6460,7358,8296,9266,  \newline 10260,11270,12289,13312,14336,15360,16384,17408,  \newline 18432,19456,20479,21498,22508,23502,24472,25410,  \newline 26308,27158,27953,28688,29360,29968,30512,30992,  \newline 31408,31760,32049,32278,32452,32578,32664,32718,  \newline 32748,32762,32767,32768,32768,32768,32768,32768,  \newline 32768,32768,32767,32762,32748,32718,32664,32578,  \newline 32452,32278,32049,31760,31408,30992,30512,29968,  \newline 29360,28688,27953,27158,26308,25410,24472,23502,  \newline 22508,21498,20479,19456,18432,17408,16384,15360,  \newline 14336,13312,12289,11270,10260,9266,8296,7358,6460,  \newline 5610,4815,4080,3408,2800,2256,1776,1360,1008,719,  \newline 490,316,190,104,50,20,6,1\} & 2097152 & 240 \\

\hline % & 268435456  & 494 \parfillskip=0pt \tabularnewline \hline &

8 & \{1,7,27,77,181,371,687,1177,1896,2904,4264,  \newline 6040,8296,11096,14504,18584,23399,29009,35469,  \newline 42827,51123,60389,70649,81919,94208,107520,121856, \newline 137216,153600,171008,189440,208896,229375,250873, \newline 273381,296883,321355,346765,373073,400231,428184, \newline 456872,486232,516200,546712,577704,609112,640872,  \newline 672921,705199,737651,770229,802893,835611,868359,  \newline 901121,933888,966656,999424,1032192,1064960,1097728,  \newline 1130496,1163264,1196031,1228793,1261541,1294259,  \newline 1326923,1359501,1391953,1424231,1456280,1488040,  \newline 1519448,1550440,1580952,1610920,1640280,1668968,  \newline 1696921,1724079,1750387,1775797,1800269,1823771,  \newline 1846279,1867777,1888256,1907712,1926144,1943552,  \newline 1959936,1975296,1989632,2002944,2015233,2026503,  \newline 2036763,2046029,2054325,2061683,2068143,2073753,  \newline 2078568,2082648,2086056,2088856,2091112,2092888,  \newline 2094248,2095256,2095975,2096465,2096781,2096971,  \newline 2097075,2097125,2097145,2097151,2097152,2097152,  \newline 2097152,2097152,2097152,2097152,2097152,2097152,  \newline 2097151,2097145,2097125,2097075,2096971,2096781,  \newline 2096465,2095975,2095256,2094248,2092888,2091112,
%& 268435456  & 494 \parfillskip=0pt \tabularnewline \hline &  
\newline 
2088856,2086056,2082648,2078568,2073753,2068143,  \newline 2061683,2054325,2046029,2036763,2026503,2015233, \newline 2002944,1989632,1975296,1959936,1943552,1926144,  \newline 1907712,1888256,1867777,1846279,1823771,1800269,  \newline 1775797,1750387,1724079,1696921,1668968,1640280,  \newline 1610920,1580952,1550440,1519448,1488040,1456280,  \newline 1424231,1391953,1359501,1326923,1294259,1261541,  \newline 1228793,1196031,1163264,1130496,1097728,1064960,  \newline 1032192,999424,966656,933888,901121,868359,835611,  \newline 802893,770229,737651,705199,672921,640872,609112,  \newline 577704,546712,516200,486232,456872,428184,400231,  \newline 373073,346765,321355,296883,273381,250873,229375,  \newline 208896,189440,171008,153600,137216,121856,107520,  \newline 94208,81919,70649,60389,51123,42827,35469,29009,  \newline 23399,18584,14504,11096,8296,6040,4264,2904,1896,  \newline 1177,687,371,181,77,27,7,1\}
& 268435456 & 494 \\
\hline
\end{longtable}
\end{center}
\end{enumerate}

\section{Summary and Outlooks} \label{summary}

In this paper, we obtained recursive wall-crossing formulas for the Poincar\'e polynomials and the Euler characteristics of Abelian symplectic quotients associated to a complex projective space of dimension $n$ acted upon in a Hamiltonian fashion by a sub-torus $T^r$ of the natural torus $T^{n+1}$, such that $2^r=n+1$. In fact, the torus $T^r$ action is an effective action and the associated symplectic quotients $M_{\xi}$, for $\xi \in \Delta_{\text{reg}}$ as regular values of the moment map $\mu_T$, possess at most orbifold singularities. The $k$-th coefficient of the Poincar\'e polynomials, i.e. the $k$-th Betti number, effectively counts the number of $k$-dimensional cycles in the corresponding symplectic quotients.

From the physical point of view, by varying the set $\xi$ of diagonal elements of reduced density matrices for an isolated $r$-qubit system inside the moment polytope ($r$-cube) and crossing critical walls between chambers of $\Delta_{\text{reg}} \subset r$-cube, we topologically classified the associated symplectic reduced spaces. In other words, the set $\xi$ encodes local phases, or equivalently the relative phases of isolated $r$ qubits, and by varying them inside their domain in a hyper-cube, we utilized a recursive wall-crossing procedure to obtain topological invariants, i.e. Poincar\'e polynomials and Euler characteristics, of the corresponding spaces of pure multi-qubit states.

However, one has to note that generalization of the wall-crossing procedure described above and especially the Abelian Poincar\'e polynomials in Eq. \eqref{recursive-wallcrossing-formula} to non-abelian group action is highly nontrivial in the sense that non-abelian moment values $\alpha = \mu(p) \in \mu(M) \cap \mathfrak{t}^*_+$, where $\mu: M \rightarrow \mathfrak{k}^*$ is the non-abelian equivariant moment map, are not necessarily regular values and therefore, the quotients $M_{\alpha} = \mu^{-1}(\alpha)/K_{\alpha}$ are stratified symplectic spaces \cite{sjamaar1991}. Therefore, to study their topological properties one can refer to intersection cohomology $IH^*(.)$, which satisfies all the properties of the cohomology groups of non-singular complex projective spaces \cite{jeffrey2003}. Finding the topological invariants, such as Poincar\'e polynomial or Euler characteristic for the non-abelian Local Unitary group action, where $T^r$ is the maximal torus, will be studied elsewhere.

\end{document}